\documentclass[conference]{IEEEtran}

\usepackage{ifthen}
\newboolean{drft}
\setboolean{drft}{false}   
\newcommand{\ifdraft}[2]{\ifthenelse{\boolean{drft}}{#1}{#2}}
\ifdraft{
	\usepackage[margin=1in,lmargin=0.5in,rmargin=1.5in,marginparwidth=1.2in]{geometry}  
}{
}

\usepackage[utf8]{inputenc}      
\usepackage[T1]{fontenc}         
\usepackage{microtype}

\usepackage{amsmath}
\usepackage{amsfonts}
\usepackage{amssymb}
\usepackage{amsthm}
\usepackage{mathtools}
\usepackage{bm}
\usepackage{optidef}
\usepackage{cite}

\usepackage{comment}
\usepackage{calc}
\usepackage[dvipsnames]{xcolor}  
\usepackage[pdftex]{graphicx}
\usepackage[inline]{enumitem}
\usepackage{subcaption}
\usepackage{hyperref}            
\hypersetup{  
	colorlinks=true,
	urlcolor=RoyalBlue,
	citecolor=Green,
	linkcolor=Red
}

\ifdraft{
\usepackage{marginfix}  
\usepackage{sidenotes}  
\let\todo\sidenote      
\makeatletter
\RenewDocumentCommand\sidenotetext{oo+m}{
	\IfNoValueOrEmptyTF{#1}{%
		\@sidenotes@placemarginal{#2}{\textsuperscript{\{\thesidenote\}}{}~\scriptsize{#3}}%
		\refstepcounter{sidenote}%
	}{%
		\@sidenotes@placemarginal{#2}{\textsuperscript{#1}~#3}%
	}%
}
\RenewDocumentCommand\sidenotemark{o}{
	\@sidenotes@multichecker%
	\IfNoValueOrEmptyTF{#1}{%
		\@sidenotes@thesidenotemark{\{\thesidenote\}}%
	}{%
		\@sidenotes@thesidenotemark{#1}%
	}%
	\@sidenotes@multimarker%
}
\makeatother
}{
	\newcommand{\todo}[1]{}
}

\input{widebar.tex}

\newtheorem{theorem}{Theorem}

\newtheorem{definition}{Definition}

\newcommand{\sm}{\setminus}

\newcommand{\given}{\,\vert\,}
\newcommand{\suff}{\succcurlyeq}

\newcommand{\ui}[2]{UI(M : {#1} \setminus {#2})}
\newcommand{\ri}[2]{RI(M : {#1} ; {#2})}
\newcommand{\si}[2]{SI(M : {#1} ; {#2})}

\newcommand{\df}[2]{\delta(M:#1\setminus #2)}

\newcommand{\sampm}{\mathsf{M}}  
\newcommand{\sampx}{\mathsf{X}}
\newcommand{\sampy}{\mathsf{Y}}
\newcommand{\T}{\mathsf{T}}

\newcommand\indept{\protect\mathpalette{\protect\independenT}{\perp}}
\def\independenT#1#2{\mathrel{\rlap{$#1#2$}\mkern2mu{#1#2}}}

\title{Capturing and Interpreting Unique Information}

\author{%
	\IEEEauthorblockN{Praveen Venkatesh}
	\IEEEauthorblockA{Allen Institute\\
		and University of Washington\\
		Seattle, WA, USA\\
		\href{mailto:praveen.venkatesh@alleninstitute.org}{\texttt{praveen.venkatesh@alleninstitute.org}}
		\vspace{-8mm}}
	\and
	\IEEEauthorblockN{Keerthana Gurushankar}
	\IEEEauthorblockA{Department of Computer Science,\\
		Carnegie Mellon University\\
		Pittsburgh, PA, USA\\
		\href{mailto:kgurusha@andrew.cmu.edu}{\texttt{kgurusha@andrew.cmu.edu}}
		\vspace{-8mm}}
	\and
	\IEEEauthorblockN{Gabriel Schamberg}
	\IEEEauthorblockA{Department of Surgery,\\
		University of Auckland\\
		New Zealand\\
		\href{mailto:gabeschamberg@gmail.com}{\texttt{gabeschamberg@gmail.com}}
		\vspace{-8mm}}
}

\begin{document}

\maketitle

\thispagestyle{plain}
\pagestyle{plain}

\begin{abstract}
	Partial information decompositions (PIDs), which quantify information interactions between three or more variables in terms of uniqueness, redundancy and synergy, are gaining traction in many application domains.
	However, our understanding of the operational interpretations of PIDs is still incomplete for many popular PID definitions.
	In this paper, we discuss the operational interpretations of unique information through the lens of two well-known PID definitions.
	We reexamine an interpretation from statistical decision theory showing how unique information upper bounds the risk in a decision problem.
	We then explore a new connection between the two PIDs, which allows us to develop an informal but appealing interpretation, and generalize the PID definitions using a common Lagrangian formulation.
	Finally, we provide a new PID definition that is able to \emph{capture} the information that is unique.
	We also show that it has a straightforward interpretation and examine its properties.
\end{abstract}

\section{Introduction}

Partial information decompositions (PIDs) have become a popular method for understanding the information interactions between multiple random variables.
A \emph{bivariate} PID seeks to decompose the information that two variables $X$ and $Y$ convey about a message $M$, into parts that are unique to $X$, unique to $Y$, redundant to $X$ and $Y$, and synergistic~\cite{williams2010nonnegative, banerjee2018unique, bertschinger2014quantifying}.

As a simple example, consider a message $M = [M_1, M_2, M_3, M_4]$, and two variables $X = [M_1, M_3, M_4 \oplus Z]$ and $Y = [M_2, M_3, Z]$, where $M_i, Z \sim $ i.i.d. Ber$(1/2)$ and $\oplus$ represents an XOR operation between bits.
Here, $X$ has one bit of unique information about $M$, i.e., $M_1$, which is not present in $Y$.
Similarly, $Y$ has one bit of unique information about $M$, i.e., $M_2$, which is not present in $X$.
There is one bit of redundant information, i.e., $M_3$, which can be extracted from \emph{either} $X$ or $Y$ taken alone.
Finally, there is one bit of synergistic information, i.e., $M_4$: this information cannot be extracted from either $X$ or $Y$ individually, but can be recovered when both are taken \emph{together}.

PIDs have found applications in various fields, from neuroscience (where one may want to examine the interaction between stimuli, neural activity and behavioral response)~\cite{pica2017quantifying, timme2018tutorial} to financial markets~\cite{scagliarini2020synergistic}.
Recent works have also used PIDs to explain how information complexity decreases through the layers of a deep neural network~\cite{ehrlich2022partial}, as well as to develop new measures of fairness in machine learning~\cite{dutta2020information}.

Despite increasingly widespread adoption, there is still no consensus on how PIDs should be defined, or on how to operationally interpret partial information quantities (e.g., see \cite{lizier2018information, kolchinsky2022novel}).
One popular approach for operational interpretations has relied on the concept of Blackwell sufficiency from statistical decision theory.
Blackwell sufficiency is a formal way to determine whether $X$ contains all of the information that $Y$ has about $M$.
Thus, it becomes a natural basis for discussing how two variables carry information about a message.
For example, Kolchinsky~\cite{kolchinsky2022novel} uses it to operationalize measures of redundancy and ``union'' information.

Here, we restrict our attention to interpretations of unique information.
Bertschinger et al.~\cite{bertschinger2014quantifying} used Blackwell sufficiency to motivate a definition of unique information.
But their interpretation only addressed whether or not the unique information was zero or non-zero, and did not provide an interpretation for the \emph{quantity} of unique information.
More recently, Banerjee et al.~\cite{banerjee2018unique} and Rauh et al.~\cite{rauh2019unique} interpreted the quantification of unique information in terms of a secret key rate using a context from information-theoretic security.
However, such an interpretation is difficult to translate to other contexts like neuroscience, where there may not be an analog for an eavesdropper.

This paper discusses two PID definitions based on Blackwell sufficiency~\cite{banerjee2018unique, bertschinger2014quantifying}, and provides an operational interpretation of the \emph{quantity} of unique information in each case.
Extending classical results on so-called ``deficiency'' measures~\cite{torgersen1991comparison, mariucci2016cam}, and clarifying results in~\cite{banerjee2018unique}, we show that the unique information about $M$ present in $X$ w.r.t.\ $Y$ upper bounds the difference in risk attained in a decision problem, when one uses $X$ rather than $Y$ to make decisions pertaining to $M$ (Sections~\ref{sec_def_ub_risk}, \ref{sec_ri_delta_symm}, and \ref{sec_ui_tilde_ub_risk}).

We then identify a previously unrecognized connection between the aforementioned PIDs, which shows that the two definitions swap the objective and constraint in their respective optimizations (Section~\ref{sec_conn_delta_tilde_pids}).
This discovery allows us to clarify how these definitions are related to Blackwell sufficiency, and provide an informal but appealing interpretation for them (Section~\ref{sec_conn_bs_informal_int}).
Finally, we develop a novel generalization of the two PIDs, through a common Lagrangian (Section~\ref{sec_generalization}).
In the process, we also explicitly raise an issue pertaining to symmetrization of redundancy, and show how it complicates the interpretation of unique information (Sections~\ref{sec_ri_delta_symm}, \ref{sec_ri_tilde_symm}).

Lastly, in Section~\ref{sec_tmxy_pid}, we propose a new PID definition that captures the \emph{part} of $M$ that is unique in the form of a random variable.
We hinted at this PID in our previous work~\cite{gurushankar2022extracting}, without defining it or discussing its properties.
Here, we define the PID formally through redundancy symmetrization, show that it forms a valid non-negative decomposition and that it obeys intuitive bounds.
We also show that this PID definition is Blackwellian~\cite{venkatesh2022partial} when $M$, $X$ and $Y$ are jointly Gaussian.

\section{Background}
\label{sec_background}

\subsection{Notation}

\begin{itemize}[leftmargin=1em]
	\item Let $M$, $X$ and $Y$ be three random variables  with sample spaces $\sampm$, $\sampx$ and $\sampy$ respectively, and joint density $P_{MXY}$.
	\item Let $\mathcal C(\mathsf{A} \given \mathsf{B})$ denote the set of all \emph{channels} from $\mathsf{A}$ to $\mathsf{B}$, so for example, $P_{X|M} \in \mathcal C(\sampx \given \sampm)$.
	\item Let $\circ$ denote composition of channels, i.e.\ $\forall\; a \in \mathsf{A}, c \in \mathsf{C}$,
		\begin{equation*}
			(P_{A|B} \circ P_{B|C})(a \given c) \coloneqq \int_{\mathsf{B}} P_{A|B}(a \given b) \cdot P_{B|C}(b \given c) \,db.
		\end{equation*}
	\item To keep the exposition simple, we ignore any measure-theoretic nuances.
		All conditional distributions and information measures are assumed to be well-defined.
\end{itemize}

\subsection{Defining PIDs}

There are many notions of partial information decompositions: we focus here on the \emph{bivariate} case, which decomposes the information that \emph{two} variables $X$ and $Y$ have about a message $M$.
Such a PID is typically defined by a set of four functions of the joint distribution $P_{MXY}$---denoted $UI(M : X \setminus Y)$, $UI(M : Y \setminus X$), $RI(M : X ; Y)$ and $SI(M : X ; Y)$ (or $UI_X$, $UI_Y$, $RI$ and $SI$ respectively for brevity)---which satisfy the following basic equations:
\begin{align}
	I\bigl(M; (X,Y)\bigr) &= \ui{X}{Y} + \ui{Y}{X} \notag \\
						  &\hphantom{= \vphantom{UI}} + \ri{X}{Y} + \si{X}{Y}, \label{eq_pid} \\
    I\bigl(M; X\bigr) &= \ui{X}{Y} + \ri{X}{Y} \label{eq_pid_x}, \\
    I\bigl(M; Y\bigr) &= \ui{Y}{X} + \ri{X}{Y} \label{eq_pid_y}.
\end{align}
Equation~\eqref{eq_pid} implies that the total mutual information about $M$ conveyed by $X$ and $Y$ is the sum of four partial information components: one unique to $X$, one unique to $Y$, another redundant to both $X$ and $Y$, and the last which is synergistic, respectively.
Equations~\eqref{eq_pid_x} and \eqref{eq_pid_y} enforce that the individual mutual information of $X$ or $Y$ with $M$ is the sum of the redundant information and the corresponding unique information.%
\footnote{Typically, it is also assumed that the redundant and synergistic components are symmetric in $X$ and $Y$.}
These equations impose three constraints on the four partial information components, such that defining any one component suffices to specify the other three.

In this paper we discuss the operational interpretations of two existing PID definitions due to \cite{banerjee2018unique} and \cite{bertschinger2014quantifying} in Section~\ref{sec_int_delta_sim_pid}, and then introduce a new PID definition in Section~\ref{sec_tmxy_pid}.
We state here the first two definitions as defined originally, and later we present modified forms which are more interpretable.
\begin{definition}[$\delta$-PID \cite{banerjee2018unique}] \label{def_delta_pid}
	Let the (weighted output) \emph{deficiency}%
	\footnote{\emph{Deficiency} was introduced by Le Cam to quantify a departure from Blackwell \emph{sufficiency}.}
	of $Y$ with respect to $X$ about $M$ be defined as%
	\footnote{The reason for this notation is that the \emph{deficiency} of $Y$ w.r.t.\ $X$ translates to the \emph{unique information} present in $X$ and not in $Y$.}
	\begin{equation}\label{eq_deficiency}
		\df{X}{Y} \coloneqq \;\; \inf_{\mathclap{\vphantom{X^{X^X}} P_{X'|Y} \,\in\, \mathcal C(\sampx|\sampy)}} \;\; \mathbb{E}_{P_{M}}\bigl[D(P_{X|M} \,\Vert\, P_{X'|Y}\circ P_{Y|M})\bigr].
	\end{equation}
	Then, the deficiency-based redundant information about $M$ present in $X$ and $Y$ is given by
	\begin{equation} \label{eq_ri_delta}
		\begin{aligned}
			RI^\delta(M : X ; Y) \coloneqq \min\{&{} I(M ; X) - \df{X}{Y}, \\
												 &{} I(M ; Y) - \df{Y}{X}\}.
		\end{aligned}
	\end{equation}
\end{definition}
Using equations \eqref{eq_pid}--\eqref{eq_pid_y}, $RI^\delta_X$ fully determines the $\delta$-PID, i.e. $UI^\delta_X$, $UI^\delta_Y$, and $SI^\delta$.
\begin{definition}[$\sim$-PID\footnote{Also called the BROJA-PID in the literature after the authors of~\cite{bertschinger2014quantifying}.} \cite{bertschinger2014quantifying, griffith2014quantifying}] \label{def_tilde_pid}
	The unique information about $M$ present in $X$ and not in $Y$ is given by
	\begin{equation} \label{eq_ui_tilde}
		\widetilde{UI}(M : X \setminus Y) \coloneqq \min_{Q \in \Delta_P} I_Q(M ; X \given Y),
	\end{equation}
	where $\Delta_P \coloneqq \{Q_{MXY}: Q_{MX} = P_{MX},\; Q_{MY} = P_{MY}\}$ and $I_Q(\cdot \given \cdot)$ is the conditional mutual information over the joint distribution $Q_{MXY}$.\todo{What else needs to be done to make this continuous here?}
\end{definition}
As with the $\delta$-PID, equations \eqref{eq_pid}--\eqref{eq_pid_y} fully determine the remaining components of the $\sim$-PID.

\subsection{Blackwell sufficiency and Blackwellian PIDs}
\label{sec_blackwellian}

Blackwell sufficiency provides a partial order between random variables based on how informative they are about a message $M$.
This notion was used by~\cite{bertschinger2014quantifying} to provide an operational motivation for the $\sim$-PID, and also underlies the basis of the $\delta$-PID~\cite{banerjee2018unique}.
\begin{definition}[Blackwell sufficiency: $\suff_M$] \label{def_blackwell}
	We say that a channel $P_{X|M}$ is \emph{Blackwell sufficient} w.r.t.\ another channel $P_{Y|M}$ (denoted $X \suff_M Y$) if $\exists\; P_{Y'|X} \in \mathcal C(\mathsf{Y} \given \mathsf{X})$ such that
	\begin{equation}
		P_{Y'|X}\circ P_{X|M} \;=\; P_{Y|M}.
	\end{equation}
\end{definition}
Intuitively, $X \suff_M Y$ means that we can generate a new random variable $Y'$ from $X$ (using the stochastic transformation $P_{Y'|X}$) so that the effective channel from $M$ to $Y'$ is equivalent to the original channel from $M$ to $Y$.%
\footnote{Blackwell sufficiency is identical to the concept of stochastic degradedness of broadcast channels~\cite{venkatesh2022partial}.}
It was shown by Blackwell~\cite{blackwell1953equivalent} that if $X$ is Blackwell sufficient for $M$ w.r.t.\ $Y$, then it is always preferable to observe $X$ rather than $Y$, for making decisions about $M$.
This operational interpretation of Blackwell sufficiency was extended to PIDs by~\cite{bertschinger2014quantifying}:
\begin{definition}[Blackwellian PID]\label{def_blackwellian}
	A bivariate PID on $P_{MXY}$ is said to be Blackwellian if
	\begin{equation*}
		UI_X = 0 \;\Leftrightarrow\; Y \suff_M X \quad\text{and}\quad UI_Y = 0 \;\Leftrightarrow\; X \suff_M Y
	\end{equation*}
\end{definition}
This means that (for a Blackwellian PID definition) the unique information in one variable is zero only if it is always beneficial to observe the \emph{other} variable to make decisions about $M$.
Conversely, if $X$ is not Blackwell sufficient for $M$ w.r.t.\ $Y$, then $Y$ must have some unique information about $M$ that $X$ cannot access.

However, it is important to note that a Blackwellian PID is only operationally motivated to the extent of whether or not the unique information is \emph{zero}.
It does not lend an operational interpretation as to the \emph{volume} of unique information when it is non-zero.

\section{Interpreting the $\delta$- and $\sim$-PIDs}
\label{sec_int_delta_sim_pid}

\subsection{Deficiency upper bounds the difference in risk}
\label{sec_def_ub_risk}

The $\delta$-PID derives its operational interpretation directly from that of deficiency~\cite{lecam1964sufficiency, torgersen1991comparison}, upon which it is based.
The deficiency of $Y$ w.r.t.\ $X$, originally defined by Le Cam~\cite{lecam1964sufficiency}, measures how far from Blackwell sufficient $Y$ is, w.r.t.\ $X$.

Le Cam's original notion of deficiency was defined using the total variation distance, and as a worst case over realizations of $M$.
That was a frequentist context, where $M$ was a statistical parameter and not a random variable. Following Raginsky~\cite{raginsky2011shannon}, the Le Cam deficiency of $Y$ w.r.t.\ $X$ about $M$ is:
\begin{equation}
	\begin{aligned}
		\MoveEqLeft[1] \delta^{\text{LeCam}}(M: X \setminus Y) \\
		& \coloneqq \inf_{\substack{P_{X'|Y} \\ \in\, \mathcal C(X|Y)}} \sup_{m \in \sampm} \, \bigl\lVert P_{X'|Y} \circ P_{Y|M=m} - P_{X|M=m} \bigr\rVert_{TV}
	\end{aligned}
\end{equation}
The Le Cam deficiency can be interpreted as upper bounding the difference in risk (for any bounded loss function) when using $X$ rather than $Y$ to make decisions based on $M$.
We can state this formally, using the setup of a decision problem:
\begin{definition}[Decision problem] \label{def_decision_problem}
	Suppose we need to perform actions based on the value of $M$, which we cannot observe directly (e.g., we may want to estimate the value of $M$).
	We have access to either $X \sim P_{X|M}$ or $Y \sim P_{Y|M}$, which can give us information about $M$.
	The actions we take after observing either $X$ or $Y$---call these $\widehat M_X(x)$ and $\widehat M_Y(y)$ respectively---incur a bounded loss that depends on the chosen action and the value of $M$.
	Let $\mathcal L(\widehat M(\cdot), M)$ ($\lVert \mathcal L \rVert_\infty \leq 1$) be the loss function, where $\widehat M(\cdot)$ may be either $\widehat M_X(x)$ or $\widehat M_Y(y)$, depending on whether we choose to observe $X$ or $Y$.
	How do we decide whether to choose $X$ or $Y$ when we do not know $\mathcal L$?
\end{definition}

Blackwell~\cite{blackwell1953equivalent} showed that if $X \suff_M Y$, we can always attain a lower loss (on average) by choosing $X$.
What happens when Blackwell sufficiency does not hold?
Define the risk as the expected loss over either $X$ or $Y$:
\begin{equation}
	\mathcal R_m(P_{X|M}, \widehat M_X, \mathcal L) \coloneqq \mathbb E_{X \sim P_{X|M=m}} \bigl[ \mathcal L(\widehat M_X(X), m) \bigr]
\end{equation}
If Blackwell sufficiency does not hold, then the worst-case risk (over $M$) when you choose $X$ is at most that when you choose $Y$, plus the Le Cam deficiency of $X$~\cite{torgersen1991comparison, mariucci2016cam}.
In other words, for any $m$ and for any $\widehat M_Y$, there exists an $\widehat M_X$ such that%
\footnote{Recall that the \emph{deficiency in $X$} is denoted $\df{Y}{X}$, because it corresponds to the \emph{unique information in $Y$}.}
\begin{equation}
	\begin{aligned}
		\mathcal R_m(P_{X|M}, \widehat M_X, \mathcal L) &\leq \mathcal R_m(P_{Y|M}, \widehat M_Y, \mathcal L) \\
													&\hphantom{\leq \vphantom{\mathcal R}}+ \delta^{\text{LeCam}}(M : Y \sm X).
	\end{aligned}
\end{equation}

Raginsky~\cite{raginsky2011shannon} showed how alternative measures like the KL-divergence may be used in place of the total variation distance, while preserving the aforementioned risk-based operational interpretation.
In that work, Raginsky preserved the frequentist setting, taking the worst case divergence between $P_{X'|M=m}$ and $P_{X|M=m}$, over all realizations of $M$.
However, for partial information decompositions, $M$ is a random variable and thus it makes more sense to consider the \emph{expected} divergence over different values of $M$.
This is what Banerjee et al.~\cite{banerjee2018unique} did, in proposing the PID stated in Definition~\ref{def_delta_pid}.
They show that the risk-based operational interpretation extends to the new deficiency definition $\delta(M : X \sm Y)$ \cite[Prop.~8]{banerjee2018unique}, but do not extend it to the corresponding unique information.
We first state the theorem for deficiency, and show the extension in the following subsection.
\begin{theorem} \label{thm_delta_pid_op_int}
	Let the average risk be given by
	\begin{equation}
		\bar{\mathcal R}(P_{X|M}, \widehat M_X, \mathcal L) \coloneqq \mathbb E_{M,X} \bigl[ \mathcal L(\widehat M_X(X), M) \bigr]
	\end{equation}
	Then, for any $\widehat M_Y$, there exists an $\widehat M_X$ such that
	\begin{equation}
		\begin{aligned}
			\bar{\mathcal R}(P_{X|M}, \widehat M_X, \mathcal L) &\leq \bar{\mathcal R}(P_{Y|M}, \widehat M_Y, \mathcal L) \\
																&\hphantom{\leq \vphantom{R}}+ g(\delta(M : Y \sm X)),
		\end{aligned}
	\end{equation}
	where $g(\cdot)$ is a monotonically increasing function.
\end{theorem}
A proof of the above theorem is presented in Appendix~\ref{app_delta_int_proof}.

\subsection{Interpreting $UI^\delta$ after redundancy-symmetrization}
\label{sec_ri_delta_symm}

\begin{figure}[t]
	\centering
	\includegraphics[width=0.5\linewidth]{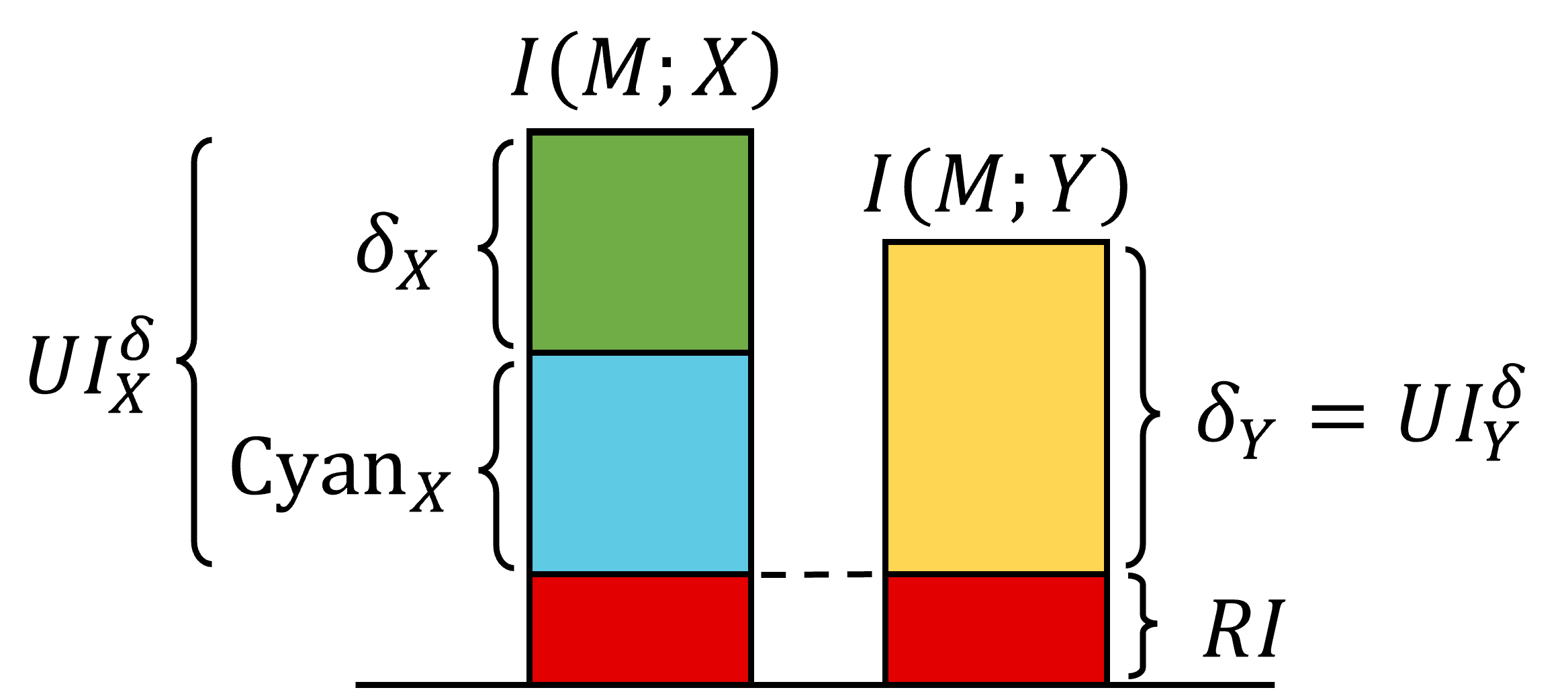}
	\caption{A depiction of the cyan region problem described in Section~\ref{sec_ri_delta_symm} for the $\delta$-PID.
	The two bars represent the quantity of mutual information $M$ has with $X$ and $Y$ respectively; the green and yellow portions represent how much of that information is the deficiency; and the red portion represents the symmetrized redundancy.
	The cyan region is part of the unique information in $X$, but cannot be accounted for by deficiency.}
	\label{fig_cyan_region}
\end{figure}

Despite the existence of a clear operational interpretation for the deficiency as defined in Definition~\ref{def_delta_pid}, the PID that arises out of this deficiency still needs an interpretation.
In particular, we need to address what happens when we symmetrize the redundancy in Equation~\eqref{eq_ri_delta}.
This symmetrization step is required because $I(M;X) - \df{X}{Y}$ is not always symmetric in $X$ and $Y$.
Interestingly, this issue does not arise in the case of the $\sim$-PID, as we discuss in Section~\ref{sec_ri_tilde_symm}.

First, we note that the operational interpretation for unique information described by Theorem~\ref{thm_delta_pid_op_int} is still valid, although the bound may be somewhat loose:
\begin{align}
	\MoveEqLeft[2] UI^\delta(M : X \setminus Y) \notag \\
	&= I(M ; X) - RI^\delta(M : X ; Y) \\
	&= \max\{\delta(M : X \setminus Y), \notag \\
	&\hphantom{= \max\{\;} \delta(M : Y \setminus X) + I(M ; X) - I(M ; Y)\} \\
	&\geq \delta(M : X \setminus Y), \label{eq_ui_gt_delta}
\end{align}
Thus, the unique information can act as an upper bound for the difference in risk, in place of deficiency.

However, one of the two unique informations, $UI^\delta_X$ or $UI^\delta_Y$, is guaranteed to be loose in this way.
We can quanitfy the extent of looseness as follows: suppose that $I(M ; X) - \df{X}{Y} > I(M; Y) - \df{Y}{X}$. Then, $RI^\delta(M : X ; Y) = I(M ; Y) - \df{Y}{X}$, and thus
\begin{align}
	UI^\delta(M : Y \sm X) &= \df{Y}{X} \\
	UI^\delta(M : X \sm Y) &= \df{Y}{X} + I(M ; X) - I(M ; Y).
\end{align}
In other words, the excess quantity added to $UI^\delta(M : X \sm Y)$, over and above the deficiency is
\begin{equation}
	\begin{aligned}
		\mathrm{Cyan}(M : X \sm Y) &\coloneqq I(M ; X) - \df{X}{Y} \\
								   &\hphantom{\coloneqq \vphantom{I}}- I(M ; Y) + \df{Y}{X}.
	\end{aligned}
\end{equation}
For lack of a better name, we call this the ``cyan region'', due to how it is depicted in Figure~\ref{fig_cyan_region}.
It is completely unclear what the interpretation of $\text{Cyan}(M : X \sm Y)$ ought to be, and why this information should be considered unique to $X$ (see Figure~\ref{fig_cyan_region}).

Essentially, we pay the cost of a loose bound in $\ui{X}{Y}$, and the extent of loosening does not have a clear justification of itself, except that it helps symmetrize the redundancy.
This gives rise to the desire for a definition that does not require the explicit symmetrization performed in Equation~\eqref{eq_ri_delta}.

\subsection{The $\sim$-PID redundancy is intrinsically symmetric}
\label{sec_ri_tilde_symm}

In a stroke of serendipity, the redundancy under the $\sim$-PID of Definition~\ref{def_tilde_pid} is naturally symmetric in $X$ and $Y$~\cite{bertschinger2014quantifying}.
Let $Q^*$ be the joint distribution that achieves the optimum in Equation~\eqref{eq_ui_tilde}. Then,
\begin{align}
	\widetilde{RI}(M : X ; Y) &= I(M ; X) - \ui{X}{Y} \\
							  &\overset{(a)}{=} I(M ; X) - I_{Q^*}(M ; X \given Y) \\
							  &\overset{(b)}{=} I_{Q^*}(M ; X) - I_{Q^*}(M ; X \given Y) \\
							  &\eqqcolon I_{Q^*}(M ; X ; Y),
\end{align}
where (a) invokes Definition~\ref{def_tilde_pid}, (b) uses the constraint that $Q^* \in \Delta_P$ so that $Q^*(m, x) = P(m, x)$, and $I_{Q^*}(M ; X ; Y)$ is the \emph{multivariate mutual information} (the negative of which is also sometimes called the \emph{interaction information}) on the distribution $Q^*$, which can be expressed as shown below (e.g., using the standard formulae from \cite[Ch.~2]{cover2012elements}):
\begin{align}
	\MoveEqLeft[1] I_{Q*}(M ; X ; Y) \\
	&= \mathbb E_{m, x, y \,\sim\, Q^*_{MXY}} \biggl[ \log \frac{Q^*(m, x) Q^*(m, y) Q^*(x, y)}{Q^*(m) Q^*(x) Q^*(y) Q^*(m, x, y)} \biggr] \notag
\end{align}
Thus, $\widetilde{RI}$ becomes equal to the multivariate mutual information on $Q^*$, which is symmetric in $x$ and $y$ by definition.

Since the $\sim$-PID has a naturally symmetric redundancy, we might want to examine whether it shares the risk-based operational interpretation of the $\delta$-PID.
We examine this, as well as alternative interpretations, in the following sections.


\subsection{$\widetilde{UI}$ upper bounds the difference in risk}
\label{sec_ui_tilde_ub_risk}

The unique information of the $\sim$-PID, $\widetilde{UI}_X$, also acts as an upper bound for the difference in risk when choosing $Y$ rather than $X$ in the decision problem from Definition~\ref{def_decision_problem}.
This follows directly from a result of Bertschinger et al.~\cite{bertschinger2014quantifying}, which states that $\widetilde{UI}_X$ upper bounds the unique information of any other PID definition that satisfies what they call ``Assumption~$(*)$''.
According to this assumption, a definition of unique information should depend only on $P_M$, $P_{X|M}$ and $P_{Y|M}$, and not on the whole joint distribution $P_{MXY}$.
Since the $\delta$-PID satisfies Assumption~$(*)$, we have that $\widetilde{UI}_X \geq UI^\delta_X$, which implies that Theorem~\ref{thm_delta_pid_op_int} extends to the $\sim$-PID as well, although the upper bound may be loose.


\subsection{A connection between the $\sim$-PID and the $\delta$-PID}
\label{sec_conn_delta_tilde_pids}


We now present a previously unidentified connection between these two PIDs, and use this connection to develop an intuitive interpretation for both PIDs.

First, observe that the $\delta$-PID can be thought of as optimizing $P_{X'|MY}$ instead of $P_{X'|Y}$, so long as we include the constraint that $M$---$Y$---$X'$ forms a Markov chain.
This constraint can also be written as $I(M ; X' \given Y) = 0$.
Thus, abbreviating $P_{X'|Y} \circ P_{Y|M}$ as $P_{X'|M}$, we can write the deficiency as:
\begin{equation} \label{eq_delta_pid_new}
	\begin{aligned}
		\MoveEqLeft[10] \delta(M : X \sm Y) = \inf_{P_{X'|MY}} \mathbb E_M\left[D_{KL}(P_{X|M} \,\Vert\, P_{X'|M})\right] \\
		&\text{s.t.} \quad I(M ; X' \given Y) = 0.
	\end{aligned}
\end{equation}

Next, we note that the definition of $\sim$-PID can also be rewritten into a similar form.
The optimization variable $Q$ in Definition~\ref{def_tilde_pid} obeys the constraints that $Q_{MX} = P_{MX}$ and $Q_{MY} = P_{MY}$.
Suppose we change notation by introducing a new random variable $X'$ using the stochastic transformation $P_{X'|MY}$, but which also obeys $P_{X'M} = P_{XM}$---or equivalently, $P_{X'|M} = P_{X|M}$.
Then, the distribution $P_{MX'Y}$ plays exactly the same role as $Q_{MXY}$, and obeys precisely the same constraints.
Thus, the $\sim$-PID definition can also be written as:
\begin{equation} \label{eq_tilde_pid_new}
	\begin{aligned}
		\MoveEqLeft[4] \widetilde{UI}(M : X \sm Y) = \inf_{P_{X'|MY}} I(M ; X' \given Y) \\
		&\text{s.t.} \quad \mathbb E_M\left[D_{KL}(P_{X|M} \,\Vert\, P_{X'|M})\right] = 0,
	\end{aligned}
\end{equation}
where the constraint $P_{X'|M} = P_{X|M}$ has been expressed in terms of zero expected KL-divergence between the channels.

This reveals the remarkable similarity between the $\delta$- and $\sim$-PIDs as written in Equations~\eqref{eq_delta_pid_new} and \eqref{eq_tilde_pid_new}.
The two PIDs are essentially optimizing over the same quantities, but in effect, interchange objective and constraint.

\subsection{Clarifying the connection to Blackwell sufficiency, and a new informal interpretation}
\label{sec_conn_bs_informal_int}

\begin{figure}[t]
	\centering
	\includegraphics[width=0.8\linewidth]{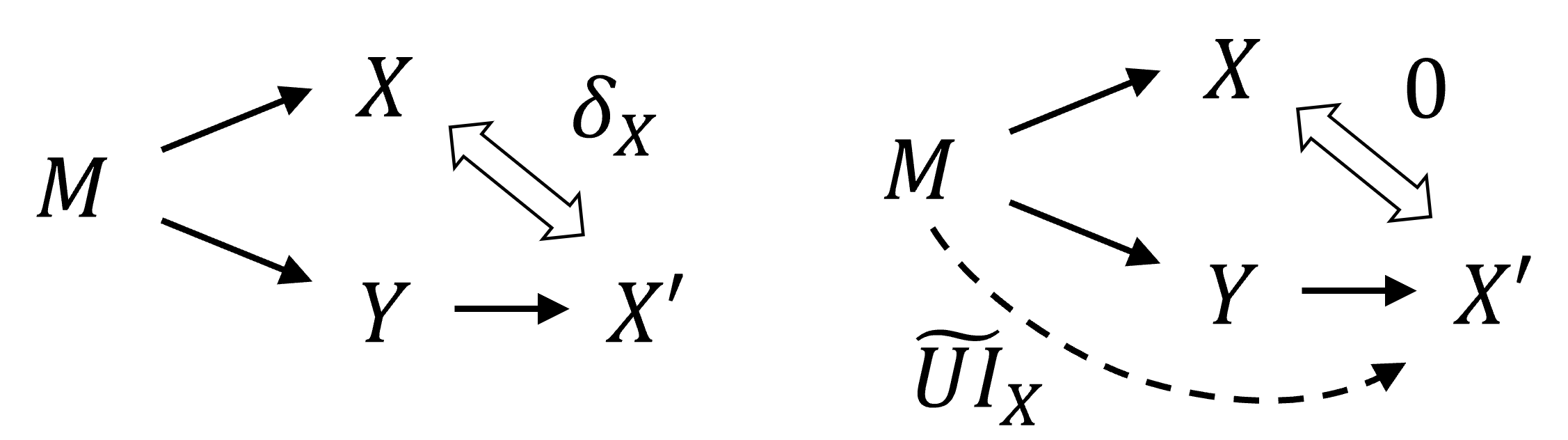}
	\caption{A depiction of the informal interpretations of the $\delta$- and $\sim$-PIDs, as described in Section~\ref{sec_conn_bs_informal_int}.
	(Left) The $\delta$-PID enforces the Markov chain $M$---$Y$---$X'$, and measures how far $X'$ is from a copy of $X$, i.e., it measures the divergence between $P_{X|M}$ and $P_{X'|M}$.
	(Right) The $\sim$-PID breaks the Markov chain by allowing bits of $M$ to leak to $X'$ outside of $Y$, however, it enforces that $X'$ is a copy of $X$.
	$\widetilde{UI}_X$ measures the minimum number of bits $X'$ needs to borrow from $M$ along the dashed line, so that $P_{X'|M}$ is a copy of $P_{X|M}$.}
	\label{fig_delta_tilde_pid_expln}
\end{figure}

Using the newfound connection between the $\delta$- and $\sim$-PIDs, we can clarify their connection to Blackwell sufficiency, and provide an informal interpretation.

First, Blackwell sufficiency can be re-understood as follows.
$Y \suff_M X$ if two requirements are met: \begin{enumerate*}[label=(\roman*)]
	\item there must exist a random variable $X'$ that is derived from $Y$ through the stochastic transformation $P_{X'|Y}$, i.e., $M$---$Y$---$X'$ must be a \emph{Markov chain}; and
	\item $X'$ must act as a ``\emph{copy}'' of $X$ w.r.t.\ $M$, in the sense that $P_{X'|M} = P_{X|M}$.%
		\footnote{This is similar to the ``simulatable'' notion presented in \cite[Defn.~38]{banerjee2018unique}.}
\end{enumerate*}

When ${Y \,\;/\,\mathclap{\suff_M}\;\;\;\, X}$, the $\delta$-PID and the $\sim$-PID quantify departures from Blackwell sufficiency in two different ways (also see Figure~\ref{fig_delta_tilde_pid_expln}): \begin{enumerate*}[label=(\roman*)]
	\item the $\delta$-PID enforces the Markov chain and measures how far we are from a copy (refer Eq.~\ref{eq_delta_pid_new});
	\item the $\sim$-PID enforces the copy and measures how far we are from having a Markov chain (refer Eq.~\ref{eq_tilde_pid_new}).
\end{enumerate*}
This unified explanation of the $\delta$- and $\sim$-PIDs has not been identified in the literature previously, to our knowledge.

We can also use this picture to offer a new informal interpretation.
If Alice and Bob opt for $X$ and $Y$ respectively in the decision problem of Definition~\ref{def_decision_problem}, the deficiency $\delta_X$ measures the closest that Bob can come to emulating Alice (on average, for the worst loss $\mathcal L$ for Bob).
On the other hand, $\widetilde{UI}_X$ measures the minimum number of bits Bob needs to borrow from $M$ in order to emulate Alice perfectly.

\subsection{A novel generalization of the $\delta$- and $\sim$-PIDs}
\label{sec_generalization}

The connection identified above also allows us to generalize both definitions using a single Lagrangian form:
\begin{equation} \label{eq_lambda_pid_new}
	\begin{aligned}
		\delta^\lambda(M : X \sm Y) \coloneqq \inf_{P_{X'|MY}} &\mathbb E_M\left[D_{KL}(P_{X|M} \,\Vert\, P_{X'|M})\right] \\
															   &+ \lambda \; I(M ; X' \given Y).
	\end{aligned}
\end{equation}
As $\lambda \to \infty$ in the Equation~\eqref{eq_lambda_pid_new}, we get the $\delta$-PID, and as $\lambda \to 0$, we get the $\sim$-PID.
This new $\delta^\lambda$-PID has to be written in terms of a deficiency and then symmetrized as in Definition~\ref{def_delta_pid}, since its redundancy will not be symmetric in general.

\section{Capturing the Unique Information}
\label{sec_tmxy_pid}

In this section, we propose a new PID definition that is able to \emph{capture} the unique information in the form of a random variable.
The quantity of unique information also has a simple operational interpretation in terms of mutual information.
\begin{definition}[$I$-PID]
    Let the \emph{information deficiency} of $Y$ with respect to $X$ about $M$ be given by
    \begin{equation}
		\delta^I(M : X \sm Y) \coloneqq \quad \sup_{\mathclap{P_{T|M} \,\in\, \mathcal C(\mathsf{T}|\sampm)}} \quad I(T ; X) - I(T ; Y).
    \end{equation}
	Here, $T$ is a random variable produced through the stochastic transformation $P_{T|M}$, and satisfies the Markov chain $T$---$M$---$(X, Y)$.
    Then, the redundant information may be defined as
    \begin{equation}
		\begin{aligned}
			RI^I(M : X ; Y) = \min\bigl\{ &I(M ; X) - \delta^I(M : X \sm Y), \\
										  &I(M ; Y) - \delta^I(M : Y \sm X) \bigr\}.
		\end{aligned}
    \end{equation}
\end{definition}

This definition is appealing, since it captures the basic intuition that if $X$ has unique information about $M$ with respect to $Y$, that means that $X$ has information about some ``part'' of $M$ which $Y$ does not have access to.
In practice, this could mean either that $X$ is able to access entire ``dimensions'' of $M$ that $Y$ cannot, or it could mean that $X$ has access to some of the same dimensions of $M$ as $Y$, but with lower noise, or it could be a combination of these factors.
In this definition, the stochastic transformation $P_{T|M}$ plays the role of \emph{extracting} these ``parts'' of $M$, which $X$ has access to but $Y$ does not.
The random variable $T$ corresponding to the optimal $P_{T|M}$ tells us the ``parts'' (or subspaces) of $M$ in which $X$ has unique information w.r.t.\ $Y$.

The operational interpretation for this definition is simply this: the unique information that $X$ has about $M$ with respect to $Y$ is the maximum information about $M$ which you can extract from $X$, which you cannot simultaneously get from $Y$.
That is, for any (possibly stochastic) function $f$ that depends only on $M$, we will always have
\begin{equation}
	I\bigl(f(M); X\bigr) \leq I\bigl(f(M) ; Y\bigr) + UI^I(M : X \sm Y).
\end{equation}
However, due to the need for symmetrization, this definition does suffer from the cyan region problem described in Section~\ref{sec_ri_delta_symm}.
This is one area where we still need to work on understanding its interpretation.

In what follows, we prove some basic properties about the $I$-PID, and show that it is Blackwellian for Gaussian $P_{MXY}$.
\begin{theorem}[Non-negativity and bounds on the $I$-PID] \label{thm_i_pid_nonnegativity}
	The $I$-PID atoms can be shown to be non-negative:
	\begin{align*}
		UI^I(M : X \sm Y) &\geq 0 & RI^I(M : X ; Y) &\geq 0 \\
		UI^I(M : Y \sm X) &\geq 0 & SI^I(M : X ; Y) &\geq 0
	\end{align*}
	The $I$-PID also satisfies the natural bounds:
	\begin{align*}
		UI^I(M : X \sm Y), \; RI^I(M : X ; Y) &\leq I(M ; X), \\
		UI^I(M : X \sm Y), \; SI^I(M : X ; Y) &\leq I(M ; X \given Y).
	\end{align*}
\end{theorem}
\begin{theorem}[The $I$-PID is Blackwellian for Gaussian $P_{MXY}$] \label{thm_i_pid_blackwellian}
	If $P_{MXY}$ is jointly Gaussian, then the $I$-PID unique information satisfies:
	\begin{equation}
		UI^I(M : X \sm Y) = 0 \quad \Leftrightarrow \quad Y \suff_M X.
	\end{equation}
\end{theorem}
Proofs of these theorems are presented in Appendix~\ref{app_i_pid_proofs}.
In particular, Theorem~\ref{thm_i_pid_blackwellian} implies that prior results for Gaussians~\cite{venkatesh2022partial} are also applicable to the $I$-PID.
We conjecture that Theorem~\ref{thm_i_pid_blackwellian} can be generalized, i.e., the $I$-PID is Blackwellian in general, but leave an investigation of this to future work.

\bibliographystyle{IEEEtran}
\bibliography{references}

\begin{thebibliography}{10}
\providecommand{\url}[1]{#1}
\csname url@samestyle\endcsname
\providecommand{\newblock}{\relax}
\providecommand{\bibinfo}[2]{#2}
\providecommand{\BIBentrySTDinterwordspacing}{\spaceskip=0pt\relax}
\providecommand{\BIBentryALTinterwordstretchfactor}{4}
\providecommand{\BIBentryALTinterwordspacing}{\spaceskip=\fontdimen2\font plus
\BIBentryALTinterwordstretchfactor\fontdimen3\font minus
  \fontdimen4\font\relax}
\providecommand{\BIBforeignlanguage}[2]{{%
\expandafter\ifx\csname l@#1\endcsname\relax
\typeout{** WARNING: IEEEtran.bst: No hyphenation pattern has been}%
\typeout{** loaded for the language `#1'. Using the pattern for}%
\typeout{** the default language instead.}%
\else
\language=\csname l@#1\endcsname
\fi
#2}}
\providecommand{\BIBdecl}{\relax}
\BIBdecl

\bibitem{williams2010nonnegative}
P.~L. Williams and R.~D. Beer, ``Nonnegative decomposition of multivariate
  information,'' \emph{arXiv preprint arXiv:1004.2515}, 2010.

\bibitem{banerjee2018unique}
P.~K. Banerjee, E.~Olbrich, J.~Jost, and J.~Rauh, ``Unique informations and
  deficiencies,'' in \emph{2018 56th Annual Allerton Conference on
  Communication, Control, and Computing (Allerton)}.\hskip 1em plus 0.5em minus
  0.4em\relax IEEE, 2018, pp. 32--38.

\bibitem{bertschinger2014quantifying}
N.~Bertschinger, J.~Rauh, E.~Olbrich, J.~Jost, and N.~Ay, ``Quantifying unique
  information,'' \emph{Entropy}, vol.~16, no.~4, pp. 2161--2183, 2014.

\bibitem{pica2017quantifying}
G.~Pica, E.~Piasini, H.~Safaai, C.~Runyan, C.~Harvey, M.~Diamond, C.~Kayser,
  T.~Fellin, and S.~Panzeri, ``Quantifying how much sensory information in a
  neural code is relevant for behavior,'' in \emph{Advances in Neural
  Information Processing Systems}, 2017, pp. 3686--3696.

\bibitem{timme2018tutorial}
N.~M. Timme and C.~Lapish, ``A tutorial for information theory in
  neuroscience,'' \emph{eneuro}, vol.~5, no.~3, 2018.

\bibitem{scagliarini2020synergistic}
T.~Scagliarini, L.~Faes, D.~Marinazzo, S.~Stramaglia, and R.~N. Mantegna,
  ``Synergistic information transfer in the global system of financial
  markets,'' \emph{Entropy}, vol.~22, no.~9, p. 1000, 2020.

\bibitem{ehrlich2022partial}
D.~A. Ehrlich, A.~C. Schneider, M.~Wibral, V.~Priesemann, and A.~Makkeh,
  ``Partial information decomposition reveals the structure of neural
  representations,'' \emph{arXiv preprint arXiv:2209.10438}, 2022.

\bibitem{dutta2020information}
S.~Dutta, P.~Venkatesh, P.~Mardziel, A.~Datta, and P.~Grover, ``An
  information-theoretic quantification of discrimination with exempt
  features,'' in \emph{Proceedings of the AAAI Conference on Artificial
  Intelligence}, vol.~34, no.~04, 2020, pp. 3825--3833.

\bibitem{lizier2018information}
J.~T. Lizier, N.~Bertschinger, J.~Jost, and M.~Wibral, ``Information
  decomposition of target effects from multi-source interactions: perspectives
  on previous, current and future work,'' p. 307, 2018.

\bibitem{kolchinsky2022novel}
A.~Kolchinsky, ``A novel approach to the partial information decomposition,''
  \emph{Entropy}, vol.~24, no.~3, p. 403, 2022.

\bibitem{rauh2019unique}
J.~Rauh, P.~K. Banerjee, E.~Olbrich, and J.~Jost, ``Unique information and
  secret key decompositions,'' in \emph{2019 IEEE International Symposium on
  Information Theory (ISIT)}.\hskip 1em plus 0.5em minus 0.4em\relax IEEE,
  2019, pp. 3042--3046.

\bibitem{torgersen1991comparison}
E.~Torgersen, \emph{Comparison of statistical experiments}.\hskip 1em plus
  0.5em minus 0.4em\relax Cambridge University Press, 1991, vol.~36.

\bibitem{mariucci2016cam}
E.~Mariucci, ``Le cam theory on the comparison of statistical models,''
  \emph{arXiv preprint arXiv:1605.03301}, 2016.

\bibitem{gurushankar2022extracting}
K.~Gurushankar, P.~Venkatesh, and P.~Grover, ``Extracting unique information
  through markov relations,'' in \emph{2022 58th Annual Allerton Conference on
  Communication, Control, and Computing (Allerton)}.\hskip 1em plus 0.5em minus
  0.4em\relax IEEE, 2022, pp. 1--6.

\bibitem{venkatesh2022partial}
P.~Venkatesh and G.~Schamberg, ``Partial information decomposition via
  deficiency for multivariate gaussians,'' in \emph{2022 IEEE International
  Symposium on Information Theory (ISIT)}.\hskip 1em plus 0.5em minus
  0.4em\relax IEEE, 2022, pp. 2892--2897.

\bibitem{griffith2014quantifying}
V.~Griffith and C.~Koch, ``Quantifying synergistic mutual information,'' in
  \emph{Guided self-organization: inception}.\hskip 1em plus 0.5em minus
  0.4em\relax Springer, 2014, pp. 159--190.

\bibitem{blackwell1953equivalent}
D.~Blackwell, ``Equivalent comparisons of experiments,'' \emph{The Annals of
  Mathematical Statistics}, pp. 265--272, 1953.

\bibitem{lecam1964sufficiency}
\BIBentryALTinterwordspacing
L.~Le~Cam, ``Sufficiency and approximate sufficiency,'' \emph{Ann. Math.
  Statist.}, vol.~35, no.~4, pp. 1419--1455, 12 1964. [Online]. Available:
  \url{https://doi.org/10.1214/aoms/1177700372}
\BIBentrySTDinterwordspacing

\bibitem{raginsky2011shannon}
M.~Raginsky, ``Shannon meets {B}lackwell and {L}e {C}am: {C}hannels, codes, and
  statistical experiments,'' in \emph{2011 IEEE International Symposium on
  Information Theory Proceedings}.\hskip 1em plus 0.5em minus 0.4em\relax IEEE,
  2011, pp. 1220--1224.

\bibitem{cover2012elements}
T.~M. Cover and J.~A. Thomas, \emph{Elements of Information Theory}.\hskip 1em
  plus 0.5em minus 0.4em\relax John Wiley \& Sons, 2012.

\bibitem{tsybakov2009introduction}
\BIBentryALTinterwordspacing
A.~B. Tsybakov, \emph{\BIBforeignlanguage{en}{Introduction to Nonparametric
  Estimation}}.\hskip 1em plus 0.5em minus 0.4em\relax New York, NY: Springer,
  2009. [Online]. Available: \url{http://dx.doi.org/10.1007/b13794}
\BIBentrySTDinterwordspacing

\end{thebibliography}

\appendices

\section{Proof of Theorem~\ref{thm_delta_pid_op_int}}
\label{app_delta_int_proof}

\begin{proof}
Consider the difference in average risks:
\begin{align}
	\MoveEqLeft[1] \bar {\mathcal R}(P_{X|M}, \hat M_X, \mathcal L) - \bar {\mathcal R}(P_{Y|M}, \hat M_Y, \mathcal L) \\
	&= \mathbb E_M \Bigl[ \mathbb E_{X|M} \bigl[ \mathcal L(\hat M_X(X), M) \bigr] - \mathbb E_{Y|M} \bigl[ \mathcal L(\hat M_Y(Y), M) \bigr] \Bigr] \notag \\
	&= \mathbb E_M \biggl[ \int P_{X|M} \cdot \mathcal L(\hat M_X(X), M) \, dx \notag \\
	&\qquad\qquad- \int P_{Y'|M} \circ P_{X|M} \cdot \mathcal L(\hat M_Y(Y), M) \, dy \notag \\
	&\qquad\qquad+ \int P_{Y'|M} \circ P_{X|M} \cdot \mathcal L(\hat M_Y(Y), M) \, dy \notag \\
	&\qquad\qquad- \int P_{Y|M} \cdot \mathcal L(\hat M_Y(Y), M) \, dy \biggr] \label{eq_risk_diff_expr}
\end{align}
Now, the last two terms of this expression can be bounded using the bound on $\mathcal L$ and the total variation distance:
\begin{align}
	\MoveEqLeft[1] \mathbb E_M \biggl[ \int P_{Y'|M} \circ P_{X|M} \cdot \mathcal L(\hat M_Y(Y), M) \, dy \notag \\
	&\qquad\qquad- \int P_{Y|M} \cdot \mathcal L(\hat M_Y(Y), M) \, dy \biggr] \\
	&= \mathbb E_M \int \bigl( P_{Y'|M} \circ P_{X|M} - P_{Y|M} \bigr) \cdot \mathcal L(\hat M_Y(Y), M) \, dy \notag \\
	&\overset{(a)}{\leq} \mathbb \lVert \mathcal L \rVert_\infty \cdot \mathbb E_M \, \bigl\lVert P_{Y'|M} \circ P_{X|M} - P_{Y|M} \bigr\rVert_{TV} \\
	&\overset{(b)}{\leq} \mathbb \lVert \mathcal L \rVert_\infty \cdot \frac{1}{\sqrt{2}} \, \mathbb E_M \, \sqrt{ D_{KL} \bigl( P_{Y|M} \big\Vert P_{Y'|M} \circ P_{X|M} \bigr) } \\
	&\overset{(c)}{\leq} \mathbb \lVert \mathcal L \rVert_\infty \cdot \sqrt{ \frac{1}{2} \, \mathbb E_M \, D_{KL} \bigl( P_{Y|M} \big\Vert P_{Y'|M} \circ P_{X|M} \bigr) } \\
	&\overset{(d)}{=} \mathbb \lVert \mathcal L \rVert_\infty \cdot g\bigl( \delta(M : Y \sm X) \bigr),
\end{align}
where in (a) we have used the bound on $\mathcal L$ and the definition of the total variation norm, in (b) we have used Pinsker's inequality~\cite[Lemma~2.5]{tsybakov2009introduction}, in (c) we have used Jensen's inequality~\cite[Thm.~2.6.2]{cover2012elements}, and in (d), we have set $g(z) \coloneqq \sqrt{z/2}$.

It only remains to be shown that the first two terms of the expression in Equation~\eqref{eq_risk_diff_expr} can be upper bounded by zero.
Examining the first two terms, for any $\hat M_Y(y)$, we can derive a stochastic action rule, $\hat M_X(x)$ that will attain the same risk: we can first draw $\tilde y \sim P_{Y'|X}$ and then select the action $\hat M_Y(\tilde y)$.
Thus,
\begin{align}
	\MoveEqLeft[1] \mathbb E_M \biggl[ \int P_{X|M} \cdot \mathcal L(\hat M_X(X), M) \, dx \\
	&\qquad\qquad- \int P_{Y'|M} \circ P_{X|M} \cdot \mathcal L(\hat M_Y(Y), M) \, dy \biggr] \leq 0, \notag
\end{align}
which completes the proof.
\end{proof}

\section{Proofs of Theorems~\ref{thm_i_pid_nonnegativity} and \ref{thm_i_pid_blackwellian}}
\label{app_i_pid_proofs}

\begin{proof}[Proof of Theorem~\ref{thm_i_pid_nonnegativity}]
	First, observe that
	\begin{align}
		\delta^I(M : X \sm Y) &\coloneqq \sup_T \; I(T ; X) - I(T ; Y) \\
							  &\geq I(0 ; X) - I(0 ; Y) = 0.
	\end{align}
	Furthermore,
	\begin{equation}
		I(T ; X) - I(T ; Y) \leq I(T ; X) \leq I(M ; X),
	\end{equation}
	where the last inequality follows by the data processing inequality and the Markov chain $T$---$M$---$(X, Y)$.
	Thus,
	\begin{align}
		\delta^I(M : X \sm Y) &\leq I(M ; X) \\
		0 \leq I(M ; X) - \delta^I(M : X \sm Y) &\leq I(M ; X) \\
		0 \leq I(M ; Y) - \delta^I(M : Y \sm X) &\leq I(M ; Y)
	\end{align}
	This implies
	\begin{gather}
		0 \leq RI^I(M : X ; Y) \leq \min \bigl\{ I(M ; X), I(M ; Y) \bigr\} \\
		0 \leq UI^I(M : X \sm Y) \leq I(M ; X) \\
		0 \leq UI^I(M : Y \sm X) \leq I(M ; Y)
	\end{gather}
	Furthermore,
	\begin{align}
		I(T ; X) - I(T ; Y) &= I(T ; (X, Y)) - I(T ; Y \given X) \\
							&\hphantom{= \vphantom{I}} - I(T ; (X, Y)) + I(T ; X \given Y) \notag \\
							&= I(T ; X \given Y) - I(T ; Y \given X) \\
							&\leq I(T ; X \given Y) \leq I(M ; X \given Y), \label{eq_cond_mi_bound}
	\end{align}
	where in the very last inequality follows from the fact that $T \indept (X, Y) \given M$ and the data processing inequality~\cite[Ch.~2]{cover2012elements}. This may not be obvious, but it follows the same proof as the data processing inequality:
	\begin{align}
		I\bigl(T, M ; X \given Y\bigr) &= I\bigl(T ; X \given Y\bigr) + I\bigl(M ; X \given Y, T\bigr) \\
		&= I\bigl(M ; X \given Y\bigr) + I\bigl(T ; X \given Y, M\bigr)
	\end{align}
	From this it follows that
	\begin{align}
		I\bigl(T ; X \given Y\bigr) + I\bigl(M ; X \given Y, T\bigr) &\overset{(a)}{=} I\bigl(M ; X \given Y\bigr) \\
		I\bigl(T ; X \given Y\bigr) &\overset{(b)}{\leq} I\bigl(M ; X \given Y\bigr),
	\end{align}
	where (a) follows from the fact that $I\bigl(T ; X \given Y, M\bigr) = 0$ since $T \indept (X, Y) \given M$, while (b) uses $I\bigl(M ; X \given Y, T\bigr) \geq 0$.
	This justifies Equation~\eqref{eq_cond_mi_bound}, which implies
	\begin{align}
		\delta^I(M : X \sm Y) &\leq I(M ; X \given Y) \\
		\delta^I(M : Y \sm X) &\leq I(M ; Y \given X)
	\end{align}
	If $UI^I(M : X \sm Y) = \delta^I(M : X \sm Y)$, then $SI^I(M : X ; Y) = I(M ; X \given Y) - \delta^I(M : X \sm Y) \geq 0$, and $SI \leq I(M ; X \given Y)$.
	This shows that all terms in the $I$-PID are non-negative and bounded.
\end{proof}

\begin{proof}[Proof of Theorem~\ref{thm_i_pid_blackwellian}]
	We need to show that when $P_{MXY}$ is jointly Gaussian,
	\begin{equation}
		UI^I_X = 0 \quad\Leftrightarrow\quad Y \suff_M X.
	\end{equation}

	$(\Leftarrow)$ Observe that the $I$-PID satisfies Assumption~$(*)$ from Bertschinger et al.~\cite{bertschinger2014quantifying}, i.e., $UI_X$ is a function only of $P_M$, $P_{X|M}$ and $P_{Y|M}$.
	Thus, by \cite[Lemma~3]{bertschinger2014quantifying}, $UI^I_X \leq \widetilde{UI}_X$.
	Since the $\sim$-PID is Blackwellian, $Y \suff_M X \;\Leftrightarrow\; \widetilde{UI}_X = 0 \Rightarrow UI^I_X = 0$.
	
	This part of the proof holds irrespective of the distribution of $P_{MXY}$.

	$(\Rightarrow)$ Now, suppose $P_{MXY}$ is Gaussian.
	Then it suffices to show that whenever ${Y \,\;/\,\mathclap{\suff_M}\;\;\;\, X}$,  $\exists\; P_{T|M}$ such that $I(T ; X) - I(T ; Y) > 0$, to ensure that $UI^I_X > 0$.

	Following the notation of \cite{venkatesh2022partial}, let $\Sigma_{MXY}$ be represent the joint covariance matrix (which fully specifies information measures on the joint distribution), let $\Sigma_{X|M}$ represent the conditional covariance matrix of $X$ given $M$ and let $\Sigma_{X,Y}$ represent the cross-covariance of $X$ and $Y$.
	Let $\Lambda_X \coloneqq \Sigma_{X,M}^\T \Sigma_{X|M}^{-1} \Sigma_{X, M}$ and $\Lambda_Y \coloneqq \Sigma_{Y,M}^\T \Sigma_{Y|M}^{-1} \Sigma_{Y, M}$.
	Then, \cite[Theorem~2]{venkatesh2022partial}, states
	\begin{equation}
		Y \suff_M X \quad\Leftrightarrow\quad \Lambda_Y \suff \Lambda_X,
	\end{equation}
	where for positive semidefinite matrices $A$ and $B$, $A \suff B$ denotes that $A - B$ is positive semidefinite.

	Consider $P_{T|M}$ to be a normal distribution, given by $\mathcal N(H_T M, \Sigma_{T|M})$.
	Further, we can assume without loss of generality that $\Sigma_M = I$.
	Then, $\Sigma_{T,X} = H_T \Sigma_M \Sigma_{M,X} = H_T \Sigma_{X,M}^\T$.
	The mutual information between $T$ and $X$ is given by:
	\begin{align}
		\MoveEqLeft[1] I(T ; X) \notag \\
		&= \frac{1}{2} \log\det(I + \Sigma_T^{-1} H^{\vphantom{\T}}_T \Sigma_{X,M}^\T \Sigma_{X|M}^{-1} \Sigma^{\vphantom{\T}}_{X,M} H_T^\T) \\
		&= \frac{1}{2} \log\det(I + \Sigma_T^{-1/2} H^{\vphantom{\T}}_T \Sigma_{X,M}^\T \Sigma_{X|M}^{-1} \Sigma^{\vphantom{\T}}_{X,M} H_T^\T \Sigma_T^{-1/2}) \\
		&= \frac{1}{2} \log\det(I + \Sigma_T^{-1/2} H^{\vphantom{\T}}_T \Lambda_X H_T^\T \Sigma_T^{-1/2})
	\end{align}
	Then,
	\begin{align}
		\df{X}{Y} &= \frac{1}{2} \log\det(I + \Sigma_T^{-1/2} H^{\vphantom{\T}}_T \Lambda_X H_T^\T \Sigma_T^{-1/2}) \\
				  &\hphantom{= \vphantom{I}}- \frac{1}{2} \log\det(I + \Sigma_T^{-1/2} H^{\vphantom{\T}}_T \Lambda_Y H_T^\T \Sigma_T^{-1/2})
	\end{align}
	If ${Y \,\;/\,\mathclap{\suff_M}\;\;\;\, X}$, then $\Lambda_Y \;\;\mathclap{\suff}\mathclap{/}\;\;\, \Lambda_X$, i.e., $\exists\; c \in \mathbb R$ s.t.
	\begin{equation}
		c^\T \Lambda_X c > c^\T \Lambda_Y c.
	\end{equation}
	Letting $\Sigma_T = I$ and $H_T = c$, we have that
	\begin{align}
		1 + c^\T \Lambda_X c &> 1 + c^\T \Lambda_Y c \\
		\det(1 + c^\T \Lambda_X c) &> \det(1 + c^\T \Lambda_Y c) \\
		\frac{1}{2} \log\det(1 + c^\T \Lambda_X c) &> \frac{1}{2} \log\det(1 + c^\T \Lambda_Y c)
	\end{align}
	This implies
	\begin{align}
		\frac{1}{2} \log\det(1 + c^\T \Lambda_X c) - \frac{1}{2} \log\det(1 + c^\T \Lambda_Y c) &> 0 \\
		\Rightarrow \qquad \df{X}{Y} &> 0
	\end{align}
	Recognizing that $UI^\delta(M : X \sm Y) \geq \delta(M : X \sm Y)$ (see Equation~\eqref{eq_ui_gt_delta}), this completes the proof.
\end{proof}

\end{document}